\def\year{2018}\relax
\documentclass[letterpaper]{article} %DO NOT CHANGE THIS
\usepackage{aaai18}  %Required
\usepackage{times}  %Required
\usepackage{helvet}  %Required
\usepackage{courier}  %Required
\usepackage{url}  %Required
\usepackage{graphicx}  %Required
\usepackage{amsmath}
\usepackage{amsthm}
\usepackage{amssymb}
\usepackage{bbold}
\usepackage{tabulary}
\usepackage{graphicx}
\usepackage{dsfont}
\usepackage{bbm}

\usepackage{multirow}
\usepackage{hyperref}

\usepackage[ruled,vlined]{algorithm2e}

\newtheorem{theorem}{Theorem}
\newtheorem{proposition}{Proposition}
\newtheorem{lemma}{Lemma}

\newtheorem{definition}{Definition}

\theoremstyle{definition}

\newenvironment{proof-sketch}{\noindent{\it PROOF (Sketch).}\hspace*{1em}}{\qed}

\newcommand{\sgn}{\operatorname{sgn}} 

\let\mathbbm\mathds

\title{Partial Truthfulness in Minimal Peer Prediction Mechanisms with Limited Knowledge
}
\author{
Goran Radanovic\\
Harvard University\\
Cambridge, USA\\
gradanovic@g.harvard.edu\\
\And
Boi Faltings\\
EPFL\\
Lausanne, Switzerland\\
boi.faltings@epfl.ch\\
}

\begin{document}

\maketitle

\begin{abstract}
We study minimal single-task peer prediction mechanisms that have limited knowledge about 
agents' beliefs. Without knowing what agents' beliefs are or
eliciting additional information, it is not possible to design a truthful mechanism in a Bayesian-Nash 
sense. We go beyond truthfulness and explore equilibrium strategy profiles that are only 
partially truthful. Using the results from the multi-armed bandit literature, we give a characterization 
of how inefficient these equilibria are comparing to truthful reporting. We measure the inefficiency of 
such strategies by counting the number of dishonest reports that any minimal knowledge-bounded 
mechanism must have. We show that the order of this number is $\Theta(\log n)$, where $n$ is the number 
of agents, and we provide a peer prediction mechanism that achieves this bound in expectation.  
\end{abstract}

\section{Introduction}

One of the crucial prerequisites for a good decision making procedure is the availability of 
accurate information, which is often distributed among many individuals. Hence, 
elicitation of distributed information represents a key component in many systems 
that are based on informative decision making. 
Typically, such an information elicitation 
scenario is modeled by representing individuals as rational agents who are willing to 
report their private information in return for (monetary) rewards. 

We study a setting in which reports cannot be directly verified, as it is the case when 
eliciting opinions regarding the outcome of a hypothetical question \cite{Garcin:2014:SOP:2892753.2892963}. 
Other examples include: product reviewing, where the reported information can describe ones 
taste, which is inherently subjective\footnote{As, for example, in rating a restaurant 
or a hotel on TripAdvisor (www.tripadvisor.com).}; peer grading, where a rater needs to grade
an essay; or eliciting information that is highly distributed, as in community sensing. 

Since a data collector cannot directly verify the reported information, it can score
submitted reports only by examining consistency among them. Such an approach is 
adopted in \textit{peer prediction mechanisms}, out of which the most known examples 
are the \textit{peer prediction method} \cite{MRZ:05} and the \textit{Bayesian truth serum} \cite{P:04}. 
While there are different ways of classifying peer prediction mechanisms, 
the most relevant one for this work distinguishes two 
categories of mechanisms by: 1) the amount of additional information they elicit from agents;
2) the knowledge they have about agents' beliefs.

The first category includes \textit{minimal} mechanisms that elicit only desired private information, 
but have some additional information that enables truthful elicitation. For instance, the classical peer 
prediction \cite{MRZ:05} assumes knowledge about how an agent forms her beliefs regarding the 
private information of other agents, while other mechanisms (e.g., see \cite{JF:11,W:14}) relax the amount 
of knowledge they need by imposing different restrictions on the agents' belief structures.    
Mechanisms from the second category elicit additional information to compensate for the lack of 
knowledge about agents' beliefs. For example, the Bayesian truth serum \cite{P:04}, and its extensions
\cite{WP:12a,RF:13,RF:14,KS:16a}, elicit agents' posterior beliefs. 

The mentioned mechanisms are designed for a \textit{single-task} elicitation scenario in which an agent's 
private information can be modeled as a sample from an unknown distribution. In a basic elicitation setting,
agents share a common belief system regarding the parameters of the setting \cite{MRZ:05,P:04}. While 
the mechanisms typically allow some deviations from this assumption \cite{FW:16,RF:14}, these deviations 
can be quite constrained, especially when private information has a complex structure.\footnote{For example, \cite{RF:14} 
show that one cannot easily relax the common prior condition when agents' private information
is real-valued.} 
 
We also mention mechanisms that operate in 
more specialized settings that allow agents to have more heterogeneous beliefs. Peer prediction 
without a common prior \cite{WP:12b} is designed for a setting in which a mechanism can clearly separate 
a period prior to agents acquiring their private information from the period after the acquisition. It elicits 
additional information from agents which corresponds to their prior beliefs.   
More recently, many mechanisms have been developed for a multi-task elicitation scenario 
\cite{DG:13,K:15,RF:16b,Shnayder:2016:EC}, primarily designed for crowdsourcing settings. 
The particularities of the multi-task setting enable the mechanisms to implicitly 
extract relevant information important for scoring agents (e.g., agents' prior beliefs).
For more details on peer prediction mechanisms, we refer the reader to \cite{FR:17}.

Clearly, there is a tradeoff between the assumed knowledge and the amount of elicited information.
The inevitability of such a trade-off can be expressed by a result of \cite{RF:13}, which states that in
the single-task elicitation setting, no minimal mechanism can achieve truthfulness for a general 
common belief system. 

\textbf{Contributions.} In this paper, we investigate single-task \textit{minimal} peer prediction 
mechanisms that have only \textit{limited} information about the agents' beliefs, which precludes 
them from incentivizing all agents to report honestly. To characterize the inefficiency 
of such an approach, we introduce a concept of \textit{dishonesty limit} that measures the minimal number 
of dishonest agents that any minimal mechanism with limited knowledge must allow. To the best of our knowledge, 
no such characterization has ever been proposed for the peer prediction setting.
Furthermore, we provide a mechanism that reaches the lower bound on the number of dishonest agents. 
Due to the fact that the bound is logarithmic in the number of reports, aggregated reports converge to
the true aggregate. 
Unlike the mechanism of \cite{JF:08,FJR:17}, that also has a goal of eliciting an accurate aggregate, 
our mechanism does not require agents to learn from each other's reports.    

The full proofs to our claims can be found in the appendix. 

\section{Formal Setting}

We study a standard peer prediction setting where agents are assumed to have a \textit{common
belief} regarding their private information \cite{MRZ:05,P:04}. 
In the considered setting, a mechanism has almost no knowledge about the agents' belief structure, 
which makes the elicitability of truthful information more challenging (e.g., \cite{RF:13}). 
We define our setting as follows. 

There are $n >> 1$ agents whose arrival to the system is stochastic. We group agents by their arrival so 
that each group has a fixed number of $k << n$ agents, and we consider participation period of a 
group as a time $t$.

To describe how agents form beliefs about their private information, we introduce a state $\omega$, which 
is a random variable that takes values in set $\Omega$, which is assumed to be a (real) interval. 
We denote the associated distribution by $p(\omega)$, and assume that $p(\omega) > 0$ for all $\omega \in \Omega$. 

An agent's private information, here called \textit{signal}, is modeled with a generic random 
variable $X$ that takes values in a finite discrete set $\{0, 1, ..., m - 1\}$ whose generic values
are denoted by $x$, $y$, $z$, etc. For each agent $i$, her signal is generated 
independently according to a distribution $\Pr(X_i|\omega)$ that depends on state variable $\omega$. This distribution is common
for agents, i.e., $\Pr(X_i|\omega) = \Pr(X_j|\omega)$ for two agents $i$ and $j$, and it is \textit{fully mixed},
i.e., for all $x \in \{0, ..., m - 1\}$ it holds that $\Pr(X = x|\omega) > 0$. 
Furthermore, we assume that private signals are \textit{stochastically relevant} \cite{MRZ:05},
meaning  that posterior distributions $\Pr(X_j | X_i = x)$ and $\Pr(X_j | X_i = y)$ (obtained from 
$\Pr(X_i|\omega)$,  $\Pr(X_j|\omega)$ and $p(\omega)$) differ for at least one value of $X_j$ whenever
$x \ne y$.  Agents share a common belief about the model parameters of the model ($\Pr(X|\omega)$ and $p(\omega)$), 
so we denote these beliefs in the same way.

Agents report their private information (signals) to a mechanism, for which they get compensations 
in terms of rewards. Agents might not be honest, so to distinguish the true signal $X$ from the reported 
one, we denote reported values by $Y$. 
Since our main result depends on agents' coordination, we also 
introduce a noise parameter that models potential imperfections in reporting strategies. In particular, 
we assume that with probability $1-\epsilon \in (0, 1]$ an agent is \textit{rational} and reports a value that 
maximizes her expected payoff, 
while otherwise she heuristically reports a random value from $\{ 0, ..., m - 1 \}$. 
Notice that we do not consider adversarial agents. 
Furthermore, while in the development of our formal results we assume that $\epsilon$ is not dependent on $X_i$, 
we also show how to apply our main mechanism when such a bias exists (Section {\em Additional Considerations}).

A mechanism needs not know the true value of $\epsilon$; it only needs to have an estimate 
$\hat \epsilon$ that is in expectation equal to $\epsilon$, and we show how to obtain $\hat \epsilon$
from the reports. Furthermore, the belief of a rational agent $i$ incorporates the fact that a peer report $Y_{j}$ 
is noisy, which means that $\Pr(Y_j = x|X_i) = (1- \epsilon) \cdot \Pr(\bar X_j = x|X_i) + \frac{\epsilon}{m}$,
where $\bar X_j$ is the value that a rational peer would report. 
 
Beliefs about an agent's signal or her report are defined on the probability simplex in $m$-dimensional space, 
that we denote by $\mathcal P$. To simplify the notation for beliefs, we often omit $X$ and $Y$ symbols.
In particular, instead of using $\Pr(X = x|\omega)$, we simply write  $\Pr(x|\omega)$, or instead of using 
$\Pr(X_j = y| X_i =  x)$, we write $\Pr(y|x)$. 

The payments of a mechanism are denoted by $\tau$ and they are applied on each group of $k$ agents separately. 
We are interested in $1$-peer payment mechanisms that reward an agent $i$ by using her one peer $j$, i.e.,
the reward function is of the form $\tau(Y_i, Y_j)$.  As shown in \cite{RF:13}, this restriction does not
limit the space of strictly incentive compatible mechanisms when agents' beliefs are not narrowed by particular 
belief updating conditions. Furthermore, we distinguish the notion of a mechanism, here denoted by $\mathcal M$, 
from a peer prediction payment function $\tau$ because different payment functions could be used on different 
groups of agents, i.e., at different time periods $t$.  

\textbf{Solution concept.} From the perspective of rational agents, our setting has a form of a Bayesian game,
hence we explore strategy profiles that are Bayesian-Nash equilibria. 
We are particularly interested in strict equilibria, in which rational agents have 
strict incentives not to deviate.   
Any mechanism that adopts honest reporting 
as a strict Bayesian-Nash equilibrium is called strictly Bayesian-Nash 
\textit{incentive compatible} (BNIC). 

\section{Our Approach}\label{sec_approach}

Let us begin by describing our approach in dealing with the impossibility of truthful minimal 
knowledge-bounded elicitation. A mechanism that we are building upon is 
described by the payment rule:
\begin{align}\label{eq_pts}
\tau(x,y,P) = d +  c \cdot \begin{cases}
\frac{1}{P(x)} &\mbox{ if } x = y\\
0 & \mbox{ otherwise }
\end{cases}
\end{align}
and is called \textit{the peer truth serum} (PTS) \cite{FPBJ:14}. $P$ is 
a fully mixed distribution that satisfies:
\begin{align}\label{eq_self_prediction}
\frac{\Pr(x|x)}{P(x)} > \frac{\Pr(y|x)}{P(y)}, y \ne x
\end{align}
Provided that other rational agents are honest, the expected payoff of a rational agent with signal 
$x$ for reporting $y$ is:
\begin{align*}
c \cdot \frac{(1-\epsilon) \cdot \Pr(y|x) + \frac{\epsilon}{m}}{P(y)}  + d
\end{align*} 
Selecting proper values for $c$ and $d$ is an orthogonal problem to the one addressed in the paper, 
and is typically achieved using a separate mechanism \cite{RF:16b}, a pre-screening process \cite{DG:13}, 
or by learning \cite{DBLP:conf/ijcai/LiuC16}. 
However, we do set $d = -c \cdot \frac{\hat \epsilon}{m\cdot P(y)}$ so that the expected payoff 
is proportional to $\frac{\Pr(y|x)}{P(y)}$ (because $\hat \epsilon = \mathbbm E(\epsilon)$). In other words,
we remove an undesirable skew in agents' expected payoffs that might occur due to the presence of 
non-strategic reports.\footnote{Furthermore, notice that by setting $c $ proportional to 
$\prod_x P(x)$, we can bound PTS payments so that they take values in $[0, 1]$.}

Without additional restrictions on agents' beliefs, 
it is possible to show that the PTS mechanism is uniquely truthful \cite{FW:16}.  
Condition \eqref{eq_self_prediction} is called the \textit{self-predicting condition}, and it is crucial for ensuring the 
truthfulness of PTS.\footnote{The condition is typically defined for $P(x)$ equal to the prior $\Pr(x)$ \cite{JF:11}, but we generalize it here.}
We say that a distribution $P$ is \textit{informative} if it satisfies the self-predicting condition. Instead of assuming that a specific \textit{a priori} 
known $P$ satisfies condition \eqref{eq_self_prediction}, we show that there always exists a certain set of distribution functions for which the 
condition holds, and although this set is initially not known, we show that one can learn it by examining the statistics of reported values for 
different reporting strategies.  

\subsection{Phase Transition Diagram}

We illustrate the reasoning behind our approach and a novel mechanism on a binary answer space $\{0, 1\}$. In this case, 
it has been shown that if we set $P(x)$ to an agent's prior belief $\Pr(x)$, the self-predicting condition is satisfied, and, consequently, the
PTS mechanism is BNIC \cite{W:14}. However, in our setting, a mechanism has no knowledge about $\Pr(x)$.

Consider what happens when $P(x)$ is much smaller than $\Pr(x)$. For signal value $y \ne x$, this means that $P(y)$ is much larger
than $\Pr(y)$. If an agent observes $x$, her expected payoff when everyone is truthful is proportional to:
\begin{align*}
\frac{\Pr(x|x)}{P(x)} > \frac{\Pr(x|x)}{\Pr(x)} >  \frac{\Pr(y|x)}{\Pr(y)}
\end{align*}
where the last inequality is due to the self-predicting condition. Therefore, agents who observe $x$ are incentivized to report it. However,
agents who observe $y$ might not be incentivized to report truthfully, because if $\frac{\Pr(x)}{P(x)} > \frac{\frac{\Pr(y|y)}{\Pr(y)}}{\frac{\Pr(x|y)}{\Pr(x)}}$,
we have: 
\begin{align*}
\frac{\Pr(x|y)}{P(x)} = \frac{\Pr(x|y)}{\Pr(x)} \cdot \frac{\Pr(x)}{P(x)}  > \frac{\Pr(y|y)}{P(y)}
\end{align*}
In this case, one would naturally expect that both observations $x$ and $y$ lead to report $x$, 
and it is easy to verify that this is an equilibrium of the PTS mechanism. Namely, the expected payoffs for 
reporting $x$ only increase if more agents report $x$. Similarly, when $P(x)$ is much larger than $\Pr(x)$, 
one would expect that agents would report $y$. 

\setlength{\abovecaptionskip}{-5pt}
\setlength{\belowcaptionskip}{-5pt}
\begin{figure}[!h]
\centering
\includegraphics[width=0.67\columnwidth]{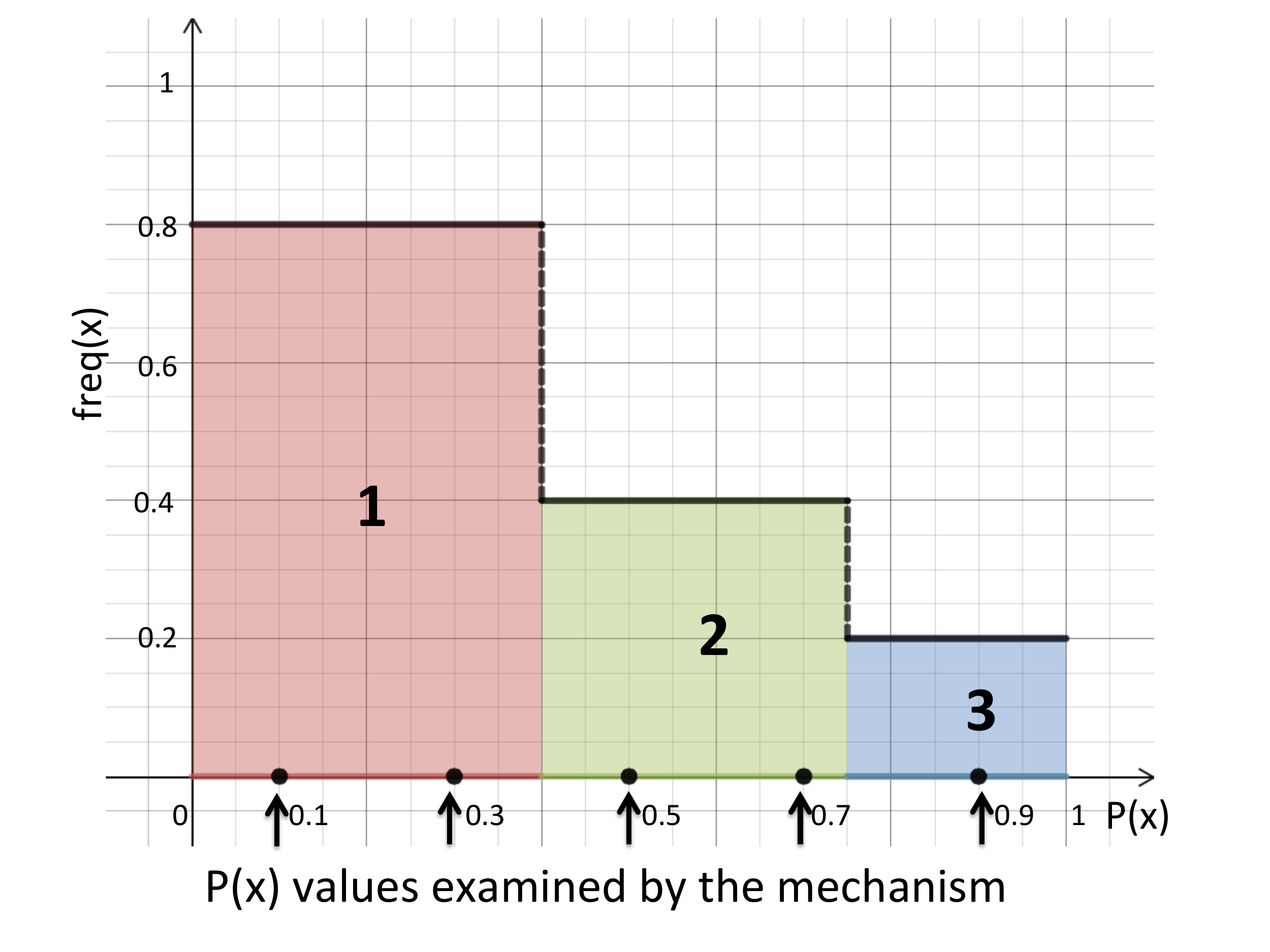}
\caption{Frequency of reports equal to $x$ for different P(x).}
\label{fig:PartialTruth_basic}
\end{figure} 

With this reasoning, we can construct a \textit{phase transition} diagram, that shows how the expected 
frequency of reports equal to $x$ changes with the increase of  $P(x)$, for a fixed posterior beliefs 
$\Pr(\cdot|x)$ and $\Pr(\cdot |y)$,  $y \ne x$. The diagram is shown in Figure \ref{fig:PartialTruth_basic}, 
and it has three phases:
\begin{itemize}
\item Phase 1, in which agents are not truthful and report $x$.  
\item Phase 2, in which agents are truthful.
\item Phase 3, in which agents are not truthful and report $y$.
\end{itemize} 
Notice that not all reports are equal to $x$ in phase 1 nor equal to $y$ in phase 3. This is due to the presence of 
noisy reports. However, noisy reports are unbiased, 
so the frequency $freq_{2}(x)$ of the truthful reporting phase is by Euclidian distance closer to $\frac{1}{2}$ than are the 
frequencies of the other two phases:
\begin{align*}
\left (\frac{1}{2} - freq_{2}(x) \right )^2 <  \left (\frac{1}{2} - freq_{i}(x) \right )^2
\end{align*} 
where $i \in \{1,3\}$, which gives us:
\begin{align*}
freq_{2}(x)\cdot \left ( 1 - freq_{2}(x) \right ) > freq(x)_{i} \cdot \left ( 1 - freq_{i}(x) \right )
\end{align*}
As the expression also holds for signal $y \ne x$, it follows that the disagreement among reports, i.e., probability that 
the two reports do not match, is (strictly) maximized in the truthful reporting phase. Therefore, we can use the disagreement 
as an indicator of whether agents are truthful or not. 

Furthermore, if $\hat \epsilon$ needs to be obtained from the reports, it is enough to 
acquire responses of $k$ agents using PTS that has $d = 0$ and $P$ such that agents are clearly 
incentivized to report a specific value. For example, one can set $P(x)$ to a small value, 
and define $\hat \epsilon = 1-num(x)/k$, where $num(x)$ is the number of agents who reported $x$. 
Notice that $\hat \epsilon$ is in expectation equal to $1-freq_1(x) = \epsilon$, and this generalizes
to the non-binary case (but also $\epsilon$ that depends on $x$ (biased noise)).

\subsection{Mechanism: Adaptive PTS (AdaPTS)}

Based on the previous observations, we now construct a novel elicitation mechanism: 
AdaPTS. The first step of the AdaPTS mechanism 
is to divide probability simplex $\mathcal P$ into regions and accordingly 
sample from each region one \textit{fully mixed} representative $P$. 
In Figure \ref{fig:PartialTruth_basic}, these representatives are shown
as black points on the horizontal axis and the division of the simplex is done uniformly. 
The granularity of the division should be fine enough so that at least one representative $P$ falls into the truthfulness 
phase. To achieve this, the PTS mechanism must have some knowledge about the agents' belief structure, but this knowledge can
be very limited. For example, to properly divide $\mathcal P$, it is enough to know the lower bound 
on the size of the region that contains distributions $P$ for which PTS is BNIC. 

Furthermore, in such a discretization, one can always choose representative $P$ distributions that are in the interior 
of the probability simplex, thus avoiding potential divisions by $0$ in equation \eqref{eq_pts}. Notice that
it is also possible to bound payments to a desired interval by choosing an appropriate value of $c$ (see Footnote 2).

Now, the AdaPTS mechanism should define payment function $\tau$ before each time step $t$, i.e., before a considered group of $k$ agents start submitting their reports.
We want to maximize the number of honest agents, without knowing for which representative distributions $P$ agents are honest. This can be translated to 
a (stochastic) multi-armed bandit (MAB) setting\footnote{A most basic $K$-armed bandit problem is defined by random variables 
$r_{i,t}$, where  $i$ represent the arm of a bandit (gambling machine) and  $r_{i,t}$ represents the reward (feedback)
obtained by pulling the arm $i$ at time step $t$. The goal is to maximize the total reward by sampling one arm at each time step. } 
(e.g., see \cite{10.2307/1427934,A:02,Audibert:2009:ETU:1519541.1519712,Audibert:2010:RBM:1756006.1953023,DBLP:journals/jmlr/GarivierC11}) with arms defined as representative distributions $P$ and the optimization goal defined as maximizing the number of honest agents. 

As argued in the previous paragraph, the latter is the same as maximizing  
the disagreement among reports. More precisely, we define our objective function 
(feedback of MAB) as an indicator function that counts the disagreements among
the reports of $k$ agents: 
\begin{align}\label{eq_disagree_fun}
I(Y_1, ..., Y_k) = \frac{1}{k \cdot (k-1)}\sum_{i} \sum_{j \ne i} \mathbb 1_{Y_i \ne Y_j}
\end{align}
Notice that the indicator function depends on a chosen representative distribution $P$, while
its associated distribution is dependent on agents' strategies and the underlying distribution 
from which agents' private signals are sampled. 
Therefore, at time step $t$, AdaPTS considers a group of $k$ agents, selects a representative $P$ distribution
according to a MAB algorithm, and scores agents using the PTS mechanism with the chosen representative 
$P$. After receiving the reports of agents, the mechanism updates the parameters of the MAB algorithm.

Although we could use any MAB algorithm with desirable regret features, in the following text we 
restrict our attention to UCB1 \cite{A:02}. Algorithm \ref{AdaPTS} depicts the pseudocode of AdaPTS based on the UCB1 algorithm 
\cite{A:02}. Function $Discretize$ returns the set of representative $P$ distributions for a given granularity $\gamma$, e.g., by uniformly 
discretizing the probability simplex as shown in Figure \ref{fig:PartialTruth_basic}.
Function $Reports$ collects the reports of $k$ agents whose rewards are then calculated using PTS mechanism with parameter $P$, 
where the peer of agent $i$ is determined by index $(i+1) \mod k$.  Function $EpsilonEstimate$ estimates the value $\epsilon$, e.g.,
by acquiring responses of $k$ agents for extremal values of $P$, as described in the previous subsection. 

\begin{algorithm}\label{AdaPTS}
\SetKwBlock{When}{when}{endwhen} 

\KwData{Granularity $\gamma > 0$, Scaling $c$;}

\Begin{

$RP = Discretize(\mathcal P, \gamma)$\;

$\hat \epsilon = EpsilonEstimate(RP, k)$\;

\For{$t = 1$ to $t = |RP|$}{
	$P = RP[t]$\;
	$\mathbf x[1:k] = Reports(P, k)$\;
	$reward[i \in [1:k]] = c \cdot \frac{\mathbbm 1_{\mathbf x[i] = \mathbf x[(i + 1) \mod k]}}{P(\mathbf x[i])} -c \cdot \frac{\hat \epsilon}{m\cdot P( \mathbf x[i])}$\;
	$I[P] = \frac{1}{k \cdot (k - 1)}\sum_{i,j \ne i \in [1:k]} \mathbbm 1_{\mathbf x[i] \ne \mathbf x[j]}$\;
	$N[P] = 1$\;
}

\For{$t = |RP| + 1$ to $t = T$}{
$P = RP \left [\arg\max_{i} \left \{ \frac{I[RP[i]]}{N[RP[i]]} + \sqrt{\frac{2\cdot \ln(t-1)}{N[RP[i]]}} \right \} \right]$\;
$\mathbf x[1:k] = Reports(P, k)$\;
$reward[i \in [1:k]] =c \cdot  \frac{\mathbbm 1_{\mathbf x[i] = \mathbf x[(i + 1) \mod k]}}{P(\mathbf x[i])} -c \cdot \frac{\hat \epsilon}{m\cdot P( \mathbf x[i])}$\;
$I[P] = I[P] + \frac{1}{k \cdot (k - 1)}\sum_{i,j \ne i \in [1:k]} \mathbbm 1_{\mathbf x[i] \ne \mathbf x[j]}$\;
$N[P] = N[P] + 1$\;
}
}
\caption{AdaPTS}
\end{algorithm}   

\section{Analysis}
 
We first start by examining particular properties of our setting, which imply the difficulty of our problem and also 
lead us towards our main results. The major technical difficulty is to show that there exists a distribution $P$ 
for which mechanism \eqref{eq_pts} is truthful. This is not a trivial statement, since the original PTS mechanism 
requires an additional condition to hold, which is not necessarily satisfied in our setting. Furthermore, 
we also need to show that \eqref{eq_disagree_fun} is an appropriate indicator function. Given these two results, 
we can apply the results from multi-armed bandit literature (\cite{LR:85,A:02}) to 
derive the logarithmic bounds on the dishonesty limit. 

\subsection{Correlation Among Signal Values}

The first property we show is that there exist a limit on how different signal values can be correlated 
in terms of agents posterior beliefs. In particular,  if an agent endorses information $x$, there is an
upper bound on the value of her belief about a peer agent endorsing information $y \ne x$. 

\begin{lemma}\label{lm_correlation}
We have:
\begin{align*}
\Pr(x|x) > \Pr(x),\forall x
\end{align*}
Furthermore, it holds that:
\begin{align*}
&\Pr(x|x) \cdot \Pr(y|y) > \Pr(y|x) \cdot \Pr(x|y), \forall y \ne x
\end{align*}
or more generally:
\begin{align*}
\Pr(x_1|x_1) \cdots \Pr(x_{m'}|x_{m'}) > \Pr(x_2|x_1) \cdots \Pr(x_1|x_{m'})
\end{align*}
where $1 < m' \le m$, $x_i \in \{0, ..., m - 1\}$, $x_i \ne x_j$ for $i \ne j$.
\end{lemma}

\subsection{Mechanisms with Limited Knowledge}\label{sec_bound_knowledge}

The second property is that the truthful elicitation of all private signals is not possible if a mechanism has no knowledge about agents' belief structure. 
This follows from Theorem 1 presented in \cite{RF:13}, which states that it is not possible to design a minimal payment scheme that truthfully elicits
private signals of all agents. While the result was obtained for a slightly different information elicitation scenario, where no 
particular belief model is assumed, it is easy to verify that it holds in our setting as well (the proof does not use anything 
contradictory to our setting). We explicitly state the impossibility of truthful information elicitation due to its importance for the further analysis. 

\begin{theorem} \label{thm_impossilibity} \cite{RF:13}
There exists no payment function $\tau$ that is BNIC for every belief model that complies with the setting.
\end{theorem}

Even if a mechanism does have some information about agents, the result of Theorem \ref{thm_impossilibity} 
is likely to hold if this knowledge is limited. We, therefore, define knowledge-bounded
mechanisms as mechanisms whose information about agents is not enough to construct a BNIC 
payments function for all admissible belief models. 

\begin{definition}\label{def_lk}
An information structure is a {\em limited knowledge} if one cannot construct a payment function 
$\tau$ that is BNIC for every belief model that complies with the setting.
\end{definition}

The AdaPTS mechanism assumes that a given granularity structure of probability 
simplex $\mathcal P$ contains a representative distribution $P$ for which the PTS 
payment rule is BNIC. We show in the following subsections that one can always partition 
probability simplex $\mathcal P$ to obtain a desirable granularity structure. Moreover, the following lemma
shows that the granularity structure of AdaPTS is not in general sufficient to construct a BNIC payment rule.

\begin{lemma}\label{lm_knowlege_pts}
The information structure that AdaPTS has about agents' beliefs is allowed to be a limited knowledge. 
\end{lemma}

\subsection{Dishonesty Limit}

While one cannot achieve incentive compatibility using peer prediction with limited knowledge, 
dishonest responses could be potentially useful for a mechanism to learn something about agents. 
This was noted in \cite{JF:08}, where the mechanism outputs a publicly available statistic that converges to the desirable outcome - true distribution of private 
signals. The drawback of the mechanism is that it relies on agents being capable of learning from each other's responses, meaning that
they update their beliefs by analyzing the changes in the public statistic.

Our approach is different. By inspecting agents' responses, we aim to learn which incentives are suitable to make agents respond 
truthfully.  To quantify what can be done with such an approach, we define \textit{dishonesty limit}.

\begin{definition}\label{def_PartialTruth}
{\em Dishonesty limit} (DL) is the minimal expected number of dishonest reports 
in any Bayesian-Nash equilibrium of any minimal mechanism with a limited knowledge. More precisely:
\begin{align*}
DL = \min_{\mathcal M} \min_{\sigma} d(\mathcal M, n, \sigma) 
\end{align*}
where $d(\mathcal M, n, \sigma)$ is the expected number of dishonest agents in an equilibrium strategy profile $\sigma$ 
of mechanism $\mathcal M$ with $n$ agents. 
\end{definition}

The following lemma establishes the order of the lower bound of DL. Its proof is based on the result of \cite{LR:85}, while the tightness of the result 
is shown in Section \textit{Main Results}. 

\begin{lemma}\label{lm_PartialTruth_lb}
The dishonesty limit is lower bounded by:
\begin{align*}
DL \ge \Omega(\log (n))  
 \end{align*}
\end{lemma} 
\begin{proof}
By Definition \ref{def_lk}, there exists no single mechanism that
incentivizes agents to report honestly if their belief model is arbitrary. Suppose now that 
we have two mechanisms $\tau_1$ and $\tau_2$ that are BNIC under two different 
(complementary) belief models, so that a particular group of agents is
truthful only for one mechanism. We can consider this situation from the perspective 
of a meta-mechanism $\mathcal M$ that has to choose between $\tau_1$ and $\tau_2$.
At a time-step $t$, $\mathcal M$ obtains $k$ reports --- feedback, which, in general, is 
insufficient for determining whether agents lied or not because agents' observations 
are stochastic, while their reports contain noise. Therefore, the problem of choosing 
between $\tau_1$ and $\tau_2$ is an instantiation of a multi-armed bandit problem
(see Section \textit{Mechanism: Adaptive PTS (AdaPTS)}). 
Since, in general, any MAB algorithm pulls suboptimal 
arms $\Omega(\log (N))$ number of times in expectation where $N$ is the total number of pulls (e.g., see \cite{LR:85}), 
we know that meta mechanism $\mathcal M$ will in expectation choose the wrong (untruthful) payments at least $O(\log n)$ times. 
This produces $O(\log n)$ untruthful reports in expectation because non-truthful payments are not truthful for at least one signal value 
and each signal value has strictly positive probability of being endorsed by an agent. 
\end{proof}

\subsection{Existence of an Informative Distribution}

The PTS mechanism is BNIC if the associated distribution 
$P$ satisfies the self-predicting condition, i.e., if $P$ is informative. 
We now turn to the crucial property of our setting: 

\begin{proposition}\label{prop_truth_exists}
There exists a region $\mathcal R \subset \mathcal P$ in probability simplex $\mathcal P$, such 
that PTS is BNIC for any $P \in \mathcal R$. 
\end{proposition}

\subsection{Indicator Function}

Now, let us define a set of lying strategies that have the same structure and in which agents report only a 
strict subset of all possible reported values. For example, if possible values are $\{0, 1, 2, 3\}$ and agents 
are not incentivized to report honestly value $3$, then a possible lying strategy
could be to report honestly values $0$, $1$, $2$, and instead of honestly reporting 
$3$, agents could report $2$.

\begin{definition}
Consider a non-surjective function $\rho : \{0, ..., m-1\} \rightarrow \{0, ..., m-1\}$. 
A reporting strategy profile is {\em non-surjective} if a report of a {\em rational} agent with private signal $X$
is $Y = \rho(X)$.   
\end{definition}
 
Non-surjective strategies also include those that most naturally follow from a simple best response reasoning: 
agents who are not incentivized to report honestly, deviate by misreporting, which necessarily reduces 
the set of values that rational agents report. This type of agents' reasoning 
basically corresponds to the {\em player inference} process explained 
in \cite{DBLP:conf/hcomp/WaggonerC14}.
Without specifying how agents form their reporting strategy, we show 
that in PTS there always exists an equilibrium non-surjective strategy profile. 
In Subsection {\em Allowing Smoother Transitions Between Phases}, we discuss how 
to use our approach when agents are not perfectly synchronized in adopting non-surjective strategies.

\begin{proposition}\label{prop_exist_non_surjective}
For any fully mixed distribution $P$, there exists a non-surjective strategy profile that is a strict equilibrium of the PTS mechanism. 
\end{proposition} 

Notice that even for non-surjective strategy profiles, the set of reported values that a mechanism receives 
does not reduce, i.e., it is equal to $\{ 0, ..., m-1\}$. 
This follows from the fact that not all agents are rational in a sense that they comply with 
a prescribed strategy profile (i.e., some report random values instead). 
Nevertheless, the statistical nature of reported values change: reports 
received by the mechanisms have smaller number of disagreements.

\begin{lemma}\label{lm_disagree}
The expected value  $\mathbbm E(I(Y_1, ..., Y_k))$ of the indicator function defined by \eqref{eq_disagree_fun} is
strictly greater for truthful reporting than for any non-surjective strategy profile. 
\end{lemma}

\subsection{Main Results}\label{sec_main}
 
From Proposition \ref{prop_truth_exists}, Proposition \ref{prop_exist_non_surjective}, and Lemma \ref{lm_disagree}, we obtain the main property of
the AdaPTS mechanism: its ability to substantially bound the number of dishonest reports in the system. 

\begin{theorem}\label{thm_num_pts_dishonest}
There exists a strict equilibrium strategy profile of the AdaPTS mechanism such that the expected number of non-truthful reports 
is of the order of $O(\log n)$.  
\end{theorem} 
\begin{proof}
Consider a reporting strategy in which agents are honest whenever $P$ is such that truthful reporting is a strict Bayesian-Nash equilibrium of PTS 
(by Proposition \ref{prop_truth_exists}, such $P$ always exists),
and otherwise they use an equilibrium non-surjective strategy profile (which by Proposition \ref{prop_exist_non_surjective} always exists). 
We use the result that the UCB1 algorithm is expected to pull a suboptimal arm $\log(N)$ times, where $N$ is the total number of pulls \cite{A:02}. 
By Lemma \ref{lm_disagree}, the representative $P$ of a truthful reporting region is an optimal arm, while the representative $P$ of 
a non-truthful region is a suboptimal arm. 
Furthermore, the number of pulls in our case corresponds to $n/k$, where $n$ is the total number of agents and $k$ is
the number of agents at time period $t$. 
Since $k << n$ is a fixed parameter, the expected number of lying agents is of the order of $O(\log(n))$.  
\end{proof}

Notice that we have not specified an exact equilibrium strategy that satisfies the bound of the theorem.
For the theorem to hold, it suffices that agents adopt truthful reporting when $P$ in AdaPTS is such that PTS is BNIC,
while they adopt any non-surjective strategy profile when $P$ is such that PTS is not BNIC. As explained in the previous 
section, a simple best response reasoning can lead to such an outcome. 

Since AdaPTS is allowed to have a bounded knowledge information structure (Lemma \ref{lm_knowlege_pts}), 
from Theorem \ref{thm_num_pts_dishonest} it follows that the dishonesty limit is upper bounded by $O(\log n)$.
From Lemma  \ref{lm_PartialTruth_lb}, we know that the dishonesty limit is lower bounded by $\Omega(\log n)$. 
Therefore:
 
\begin{theorem}\label{thm_PartialTruth}
The dishonesty limit is:
\begin{align*} 
DL  = \Theta(\log n) 
\end{align*}
\end{theorem}

\textbf{Importance of the results.} With the dishonesty limit concept, we are able to quantify what is possible 
in the context of minimal elicitation with limited knowledge. An example of an objective that is possible to reach with 
partially truthful mechanisms is elicitation of accurate aggregates. In particular, suppose that the goal is to elicit 
a distribution of signal values. From Theorem \ref{thm_num_pts_dishonest}, we know that this is achievable with AdaPTS. 
Namely, if we denote the normalized histogram of reports by $H_{reports}$  and  the normalized 
histogram of signal values by $H_{values}$, then it follows from the theorem that their expected difference is bounded by:
\begin{align*}
\mathbbm E(\sum_x |H_{reports}(x) - H_{values}(x)|) \le \frac{O(\log n)}{n}
\end{align*}
which approaches $0$ as $n$ increases. Therefore, although the truthfulness of all agents is not
guaranteed, the aggregate obtain by the partially truthful elicitation converges to the one that would 
be obtained if all agents were honest. 

\section{Additional Considerations}

In this section, we discuss how to make our mechanism AdaPTS applicable to the situations when there exist biases in 
reporting errors (i.e., $\epsilon$ parameter is biased towards a particular value) or the phase transitions are smoother (e.g., 
because reporting strategies are not perfectly synchronized). 

\textbf{Allowing Bias in the Non-strategic Reports. } 
Allowing biases in reporting errors is important for cases when some agents do not strategize w.r.t. parameter $P$, e.g.,
these agents are truthful regardless of the payment function, or they report heuristically without observing their private 
signal (for example, a fraction of agents reports $0$, whereas the other agents are strategic).
An example of the phase transition diagram that incorporates a bias is shown in Figure \ref{fig:PartialTruth_ext2}. 
Since the disagreement is highest for phase 1 of the diagram, function \eqref{eq_disagree_fun} is not a good 
indicator of agents' truthfulness.  However, frequency $freq_{2}(x)$ of the truthful phase is the closest one to the average of 
$f^{m}(x) = freq_{1}(x)$ and $f_{m}(x) = freq_{3}(x)$. By the same reasoning as in Section \textit{Our Approach}, %\ref{sec_approach}, 
the expression $(f^{m}(x) - freq_{i}(x)) \cdot (freq_{i}(x) - f_{m}(x))$
is maximized for $i = 2$. 
This leads us the following indicator function:
\begin{align*}
I (Y_1, ...., Y_k ) = \sum_{x} &(f^{m}(x) - \hat f(x)) \cdot (\hat f(x) - f_{m}(x))
\end{align*}
where $\hat f$ is the frequency of reports equal to $x$ among $\{Y_1, ..., Y_k\}$.
The problem, however, is that the AdaPTS mechanism does not know $f^m$ and $f_m$. 
Nevertheless, it can estimate them in an online manner from agents' reports. 

To adjust AdaPTS, we first define time intervals $\theta_i = \{ 2^{i}, ... ,2^{i+1}\}$ and to each 
time interval associate a different UCB1 algorithm. For periods $t \in \theta_i$, we run a (separate) 
UCB1 algorithm 
with the indicator function  
that uses estimators $\hat f_m(x)$ and $\hat f^m(x)$. The estimators are initially set to $\hat f_m(x) = 0$ 
and $\hat f^m(x) = 1$ for all values $x$, and they change after each time interval $\theta_i$. More precisely,
they are updated by finding respectively the minimum and the maximum frequency of reports equal to $x$ 
among all possible arms (representative distributions $P$). Since UCB1 sufficiently explores suboptimal arms,
the estimates $\hat f^m(x)$ and $\hat f_m(x)$ become reasonable accurate at some point, which implies a 
sublinear number of dishonest reports for a longer elicitation period.

\setlength{\abovecaptionskip}{-5pt}
\setlength{\belowcaptionskip}{-5pt}
\begin{figure}[!h]
\centering
\includegraphics[width=0.67\columnwidth]{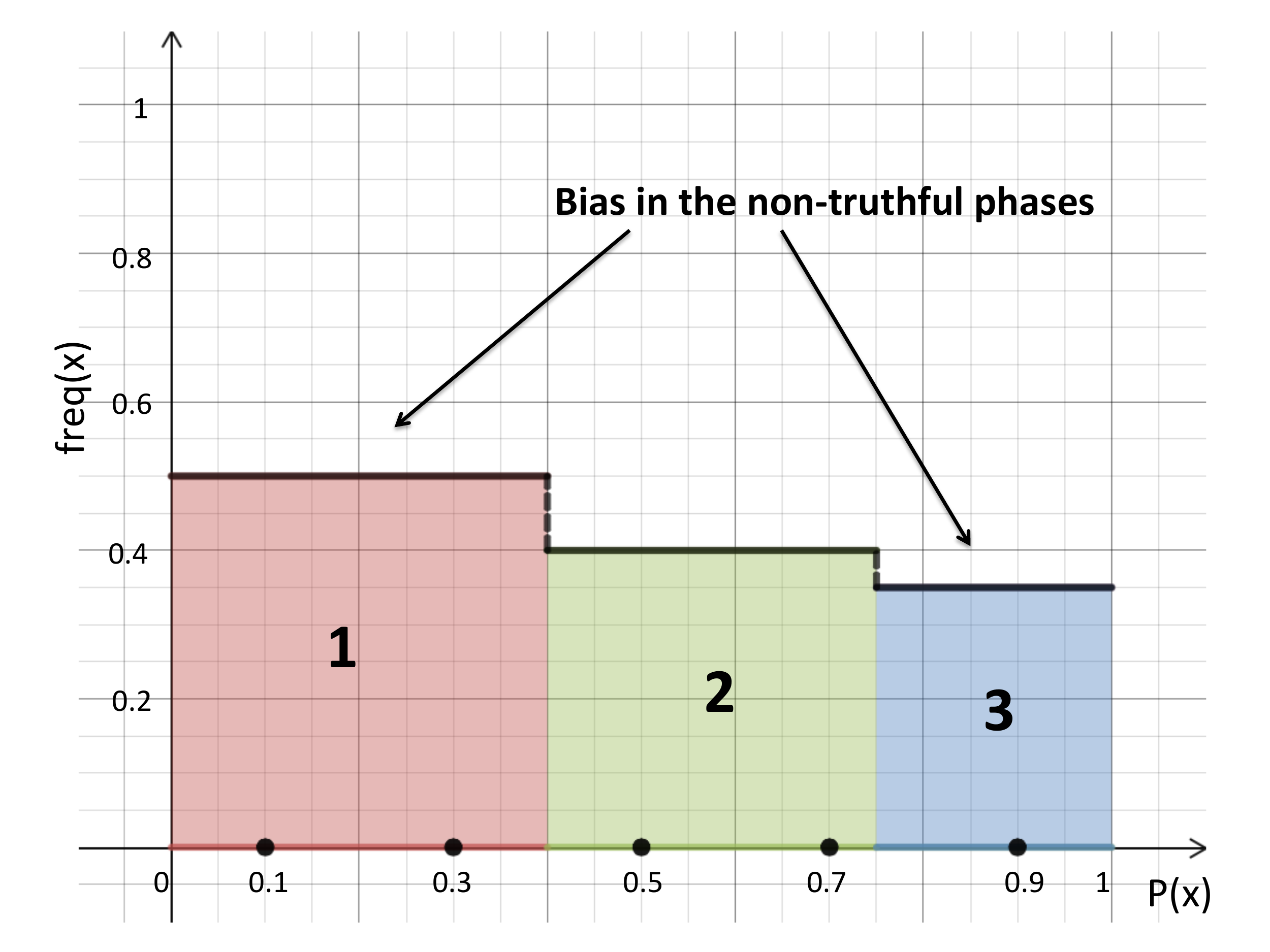}
\caption{Phase transition diagram with biased phases.}
\label{fig:PartialTruth_ext2}
\end{figure} 

\textbf{Allowing Smoother Transitions Between Phases. }
Agents may not be perfectly synchronized in changing between phases. Nonetheless, we can expect that the resulting behaviour would produce a similar 
phase transition diagram, as illustrated in Figure \ref{fig:PartialTruth_ext1}. However, the simple indicator function defined by expression \eqref{eq_disagree_fun} 
is no longer a suitable choice for detecting the truthful reporting phase.   

\setlength{\abovecaptionskip}{-5pt}
\setlength{\belowcaptionskip}{-5pt}
\begin{figure}[!h]
\centering
\includegraphics[width=0.67\columnwidth]{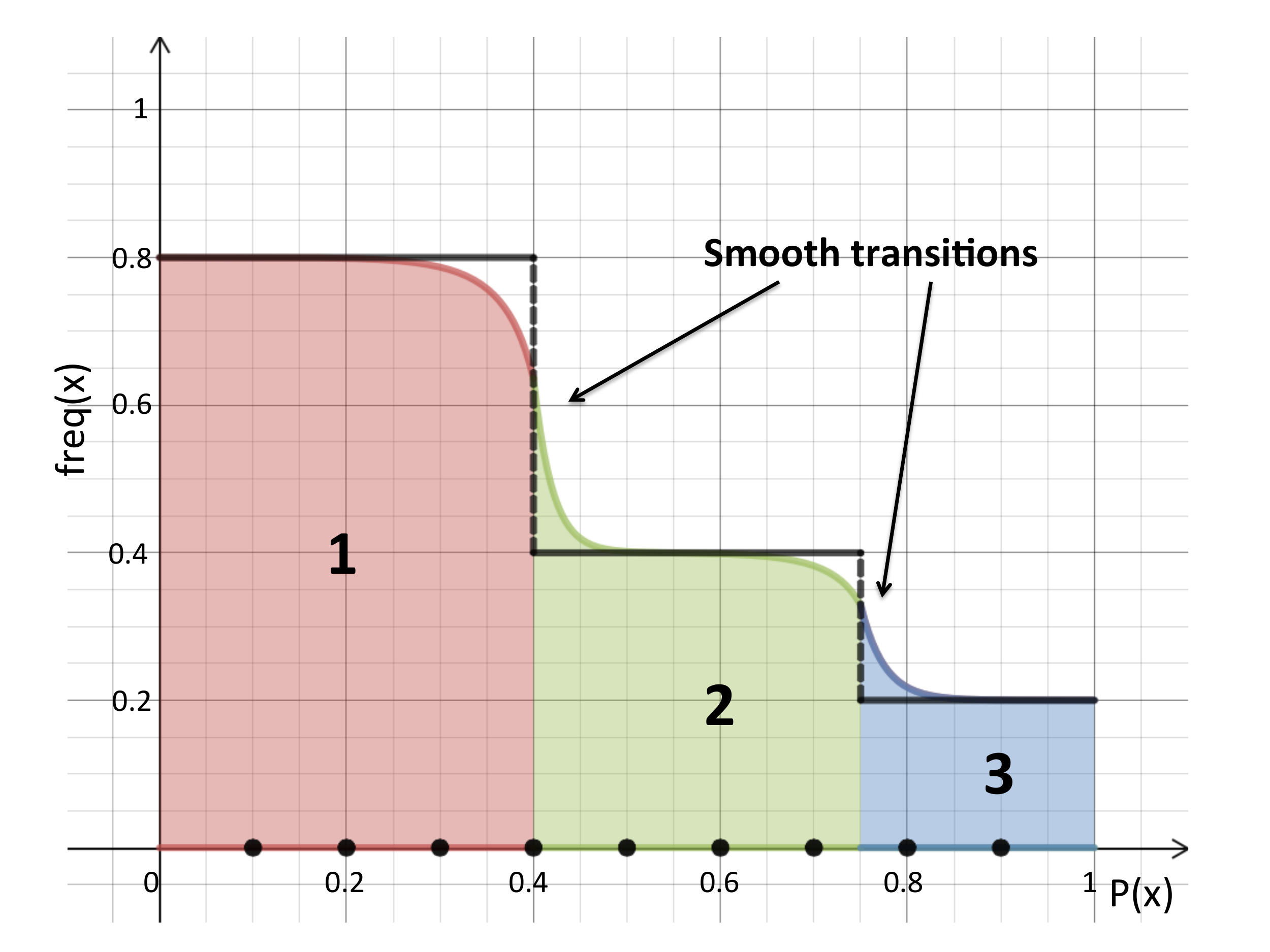}
\caption{Phase transition diagram with smooth transitions.}
\label{fig:PartialTruth_ext1}
\end{figure} 

In order to see this, we have added additional representative distributions $P$. Notice that we obtain the same disagreement 
for $P(x) = 0.4$ as for $P(x) = 0.6$. However, $P(x) = 0.6$ is a better choice,
because $P(x) = 0.4$ belongs to a transition phase where a high level of disagreement is due to asynchronous behaviour of agents. 
Notice that the phase diagram experiences rapid changes in transitions between two phases. This means  
that we can avoid selection of undesirable $P$ distributions by introducing a proper regularization term. That is, we can separate $k$
agents that arrive at time $t$ into two groups, $G_1$ and $G_2$, and reward each group with a slightly different $P$ from the one selected by UCB1. 
For example, if $P(x) = 0.4$ is selected, we could reward one group using PTS with $P(x) = 0.35$ and the other group using PTS with $P(x) = 0.45$. 
If group $G_1$ has agents $\{a_1,..., a_4\}$ and group $G_2$ has agents $\{b_1, ..., b_4\}$, then a possible indicator function could be: 
\begin{align*}
&I (Y_1, ...., Y_k ) = \underbrace{\sum_{a_i,a_j \in G_1} \mathbbm 1_{Y_{a_i} \ne Y_{a_j}} +\sum_{b_i,b_j \in G_2} \mathbbm 1_{Y_{b_i} \ne Y_{b_j}}}_{\text{disagreement}} \\
&-  \underbrace{\lambda \cdot (\mathbbm 1_{Y_{a_1} \ne Y_{a_2}} - \mathbbm 1_{Y_{b_1} \ne Y_{b_2}}) \cdot (\mathbbm 1_{Y_{a_3} \ne Y_{a_4}} -  \mathbbm 1_{Y_{b_3} \ne Y_{b_4}})}_{\textit{regularization}}
\end{align*}
where $\lambda > 0$ is the regularization factor. The regularization term is in expectation equal to 
the square of the difference between the expected disagreement of agents in group $G_1$ and 
the expected disagreement of agents in group $G_2$. Some insight on how to adjust $\lambda$
might be a priori needed, but this information is a limited knowledge. With this modification of the 
indicator function, we can apply AdaPTS. 

\section{Conclusion}

We investigated the asymptotic behavior of partially truthful minimal peer prediction mechanisms 
with a limited knowledge. As shown by Theorem \ref{thm_PartialTruth}, 
any such mechanism results in $O(\log n)$ {\em redundant} (non-truthful) reports. 
In contrast, one of the most known knowledge-bonded elicitation mechanism, Bayesian Truth Serum \cite{P:04}, 
elicits from each agent her signal value and her prediction about other agents, having in total $O(n)$ additional 
reports.  
Thus, our results quantify the necessary overhead when the minimality in reported information and the knowledge of 
a mechanism is preferred to full truthfulness. One of the most important future steps would be to make the 
mechanism robust in terms of collusion resistance (e.g., measured using replicator dynamics \cite{Shnayder:2016:IJCAI}),
which, in general, can be challenging even for a more robust settings \cite{DBLP:journals/corr/GaoWL16}.

\vspace{0.5mm}
{\bfseries Acknowledgments}
This work was supported in part by the Swiss National Science Foundation (Early Postdoc Mobility fellowship).

\bibliographystyle{aaai}
\bibliography{PartialTruth}

\begin{thebibliography}{}

\bibitem[\protect\citeauthoryear{Agrawal}{1995}]{10.2307/1427934}
Agrawal, R.
\newblock 1995.
\newblock Sample mean based index policies with o(log n) regret for the
  multi-armed bandit problem.
\newblock {\em Advances in Applied Probability} 27(4):1054--1078.

\bibitem[\protect\citeauthoryear{Audibert and
  Bubeck}{2010}]{Audibert:2010:RBM:1756006.1953023}
Audibert, J.-Y., and Bubeck, S.
\newblock 2010.
\newblock Regret bounds and minimax policies under partial monitoring.
\newblock {\em J. Mach. Learn. Res.} 11:2785--2836.

\bibitem[\protect\citeauthoryear{Audibert, Munos, and
  Szepesv\'{a}ri}{2009}]{Audibert:2009:ETU:1519541.1519712}
Audibert, J.-Y.; Munos, R.; and Szepesv\'{a}ri, C.
\newblock 2009.
\newblock Exploration-exploitation tradeoff using variance estimates in
  multi-armed bandits.
\newblock {\em Theor. Comput. Sci.} 410(19):1876--1902.

\bibitem[\protect\citeauthoryear{Auer, Cesa-Bianchi, and Fischer}{2002}]{A:02}
Auer, P.; Cesa-Bianchi, N.; and Fischer, P.
\newblock 2002.
\newblock Finite-time analysis of the multiarmed bandit problem.
\newblock {\em Machine Learning} 47(2-3):235--256.

\bibitem[\protect\citeauthoryear{Dasgupta and Ghosh}{2013}]{DG:13}
Dasgupta, A., and Ghosh, A.
\newblock 2013.
\newblock Crowdsourced judgement elicitation with endogenous proficiency.
\newblock In {\em Proceedings of the 22nd ACM International World Wide Web
  Conference (WWW'13)}.

\bibitem[\protect\citeauthoryear{Faltings and Radanovic}{2017}]{FR:17}
Faltings, B., and Radanovic, G.
\newblock 2017.
\newblock {\em Game Theory for Data Science: Eliciting Truthful Information}.
\newblock Morgan \& Claypool Publishers.

\bibitem[\protect\citeauthoryear{Faltings \bgroup et al\mbox.\egroup
  }{2014}]{FPBJ:14}
Faltings, B.; Pu, P.; Tran, B.~D.; and Jurca, R.
\newblock 2014.
\newblock Incentives to counter bias in human computation.
\newblock In {\em Proceedings of the Second {AAAI} Conference on Human
  Computation and Crowdsourcing (HCOMP'14)}.

\bibitem[\protect\citeauthoryear{Faltings, Jurca, and Radanovic}{2017}]{FJR:17}
Faltings, B.; Jurca, R.; and Radanovic, G.
\newblock 2017.
\newblock Peer truth serum: Incentives for crowdsourcing measurements and
  opinions.
\newblock {\em CoRR} abs/1704.05269.

\bibitem[\protect\citeauthoryear{Frongillo and Witkowski}{2016}]{FW:16}
Frongillo, R., and Witkowski, J.
\newblock 2016.
\newblock A geometric method to construct minimal peer prediction mechanisms.
\newblock In {\em Proceedings of the 30th AAAI Conference on Artificial
  Intelligence (AAAI'16)}.

\bibitem[\protect\citeauthoryear{Gao, Wright, and
  Leyton{-}Brown}{2016}]{DBLP:journals/corr/GaoWL16}
Gao, A.; Wright, J.~R.; and Leyton{-}Brown, K.
\newblock 2016.
\newblock Incentivizing evaluation via limited access to ground truth:
  Peer-prediction makes things worse.
\newblock {\em CoRR} abs/1606.07042.

\bibitem[\protect\citeauthoryear{Garcin and
  Faltings}{2014}]{Garcin:2014:SOP:2892753.2892963}
Garcin, F., and Faltings, B.
\newblock 2014.
\newblock Swissnoise: Online polls with game-theoretic incentives.
\newblock In {\em Proceedings of the Twenty-Eighth AAAI Conference on
  Artificial Intelligence (AAAI'14)}.

\bibitem[\protect\citeauthoryear{Garivier and
  Capp{\'{e}}}{2011}]{DBLP:journals/jmlr/GarivierC11}
Garivier, A., and Capp{\'{e}}, O.
\newblock 2011.
\newblock The {KL-UCB} algorithm for bounded stochastic bandits and beyond.
\newblock In {\em The 24th Annual Conference on Learning Theory (COLT'11)},
  359--376.

\bibitem[\protect\citeauthoryear{Jurca and Faltings}{2008}]{JF:08}
Jurca, R., and Faltings, B.
\newblock 2008.
\newblock Incentives for expressing opinions in online polls.
\newblock In {\em Proceedings of the 9th ACM conference on Electronic commerce
  (EC'08)}.

\bibitem[\protect\citeauthoryear{Jurca and Faltings}{2011}]{JF:11}
Jurca, R., and Faltings, B.
\newblock 2011.
\newblock Incentives for answering hypothetical questions.
\newblock In {\em Workshop on Social Computing and User Generated Content}.

\bibitem[\protect\citeauthoryear{Kamble \bgroup et al\mbox.\egroup
  }{2015}]{K:15}
Kamble, V.; Shah, N.; Marn, D.; Parekh, A.; and Ramachandran, K.
\newblock 2015.
\newblock Truth serums for massively crowdsourced evaluation tasks.
\newblock In {\em the 5th Workshop on Social Computing and User Generated
  Content}.

\bibitem[\protect\citeauthoryear{Kong and Schoenebeck}{2016}]{KS:16a}
Kong, Y., and Schoenebeck, G.
\newblock 2016.
\newblock Equilibrium selection in information elicitation without verification
  via information monotonicity.
\newblock {\em CoRR} abs/1603.07751.

\bibitem[\protect\citeauthoryear{Lai and Robbins}{1985}]{LR:85}
Lai, T.~L., and Robbins, H.
\newblock 1985.
\newblock Asymptotically efficient adaptive allocation rules.
\newblock {\em Advances in Applied Mathematics} 6(1):4--22.

\bibitem[\protect\citeauthoryear{Liu and Chen}{2016}]{DBLP:conf/ijcai/LiuC16}
Liu, Y., and Chen, Y.
\newblock 2016.
\newblock Learning to incentivize: Eliciting effort via output agreement.
\newblock In {\em Proceedings of the 25th International Joint Conference on
  Artificial Intelligence (IJCAI'16)}.

\bibitem[\protect\citeauthoryear{Miller, Resnick, and
  Zeckhauser}{2005}]{MRZ:05}
Miller, N.; Resnick, P.; and Zeckhauser, R.
\newblock 2005.
\newblock Eliciting informative feedback: The peer-prediction method.
\newblock {\em Management Science} 51:1359--1373.

\bibitem[\protect\citeauthoryear{Prelec}{2004}]{P:04}
Prelec, D.
\newblock 2004.
\newblock A bayesian truth serum for subjective data.
\newblock {\em Science} 34(5695):462--466.

\bibitem[\protect\citeauthoryear{Radanovic and Faltings}{2013}]{RF:13}
Radanovic, G., and Faltings, B.
\newblock 2013.
\newblock A robust bayesian truth serum for non-binary signals.
\newblock In {\em Proceedings of the 27th AAAI Conference on Artificial
  Intelligence (AAAI'13)}.

\bibitem[\protect\citeauthoryear{Radanovic and Faltings}{2014}]{RF:14}
Radanovic, G., and Faltings, B.
\newblock 2014.
\newblock Incentives for truthful information elicitation of continuous
  signals.
\newblock In {\em Proceedings of the 28th AAAI Conference on Artificial
  Intelligence (AAAI'14)}.

\bibitem[\protect\citeauthoryear{Radanovic, Faltings, and Jurca}{2016}]{RF:16b}
Radanovic, G.; Faltings, B.; and Jurca, R.
\newblock 2016.
\newblock Incentives for effort in crowdsourcing using the peer truth serum.
\newblock {\em ACM Transactions on Intelligent Systems and Technology (TIST)}
  7:48:1--48:28.

\bibitem[\protect\citeauthoryear{Shnayder \bgroup et al\mbox.\egroup
  }{2016a}]{Shnayder:2016:EC}
Shnayder, V.; Agarwal, A.; Frongillo, R.; and Parkes, D.~C.
\newblock 2016a.
\newblock Informed truthfulness in multi-task peer prediction.
\newblock In {\em Proceedings of the 2016 ACM Conference on Economics and
  Computation (EC'16)}.

\bibitem[\protect\citeauthoryear{Shnayder \bgroup et al\mbox.\egroup
  }{2016b}]{Shnayder:2016:IJCAI}
Shnayder, V.; Agarwal, A.; Frongillo, R.; and Parkes, D.~C.
\newblock 2016b.
\newblock Measuring performance of peer prediction mechanisms using replicator
  dynamics.
\newblock In {\em In Proceedings of the 25th International Joint Conference on
  Artificial Intelligence (IJCAI'16)}.

\bibitem[\protect\citeauthoryear{Waggoner and
  Chen}{2014}]{DBLP:conf/hcomp/WaggonerC14}
Waggoner, B., and Chen, Y.
\newblock 2014.
\newblock Output agreement mechanisms and common knowledge.
\newblock In {\em Proceedings of the Second {AAAI} Conference on Human
  Computation and Crowdsourcing (HCOMP'14)}.

\bibitem[\protect\citeauthoryear{Witkowski and Parkes}{2012a}]{WP:12b}
Witkowski, J., and Parkes, D.~C.
\newblock 2012a.
\newblock Peer prediction without a common prior.
\newblock In {\em Proceedings of the 13th ACM Conference on Electronic Commerce
  (EC'12)}.

\bibitem[\protect\citeauthoryear{Witkowski and Parkes}{2012b}]{WP:12a}
Witkowski, J., and Parkes, D.~C.
\newblock 2012b.
\newblock A robust bayesian truth serum for small populations.
\newblock In {\em Proceedings of the 26th AAAI Conference on Artificial
  Intelligence (AAAI'12)}.

\bibitem[\protect\citeauthoryear{Witkowski}{2014}]{W:14}
Witkowski, J.
\newblock 2014.
\newblock {\em Robust Peer Prediction Mechanisms}.
\newblock Ph.D. Dissertation, Albert-Ludwigs-Universitat Freiburg: Institut fur
  Informatik.

\end{thebibliography}

\section*{ATTACHMENT: Partial Truthfulness in Peer Prediction Mechanisms with Limited Knowledge}

\subsection{Proof of Lemma \ref{lm_correlation} }

\begin{proof}
Using the properties of our model (conditional independence of signal values 
given $\omega$) and Bayes' rule we obtain:
\begin{align*}
\Pr(x|x) &= \int_{\omega} \Pr(x|\omega) \cdot  p(\omega|x) d \omega = \int_{\omega} \Pr(x|\omega)^2 \cdot  \frac{p(\omega)}{\Pr(x)} d \omega \\
&= \frac{\int_{\omega} \Pr(x|\omega)^2 \cdot  p(\omega) d \omega}
	  {\int_{\omega} \Pr(x|\omega) \cdot  p(\omega) d \omega} 
\end{align*}

Jensen's inequality tells us that $\int_{\omega} \Pr(x|\omega)^2 \cdot  p(\omega) d \omega \ge 
(\int_{\omega} \Pr(x|\omega) \cdot  p(\omega) d \omega)^2$, with strict inequality when 
$0< Pr(x|\omega) < 1$ (notice that $Pr(x|\omega)$ is not constant due to stochastic relevance). As $Pr(x|\omega)$ is fully mixed, we have:
\begin{align*}
\Pr(x|x) &> \frac{\left ( \int_{\omega} \Pr(x|\omega) \cdot  p(\omega) d \omega \right )^2}
	  {\int_{\omega} \Pr(x|\omega) \cdot  p(\omega) d \omega} >  \int_{\omega} \Pr(x|\omega) \cdot  p(\omega) d \omega \\
	  &= Pr(x)
\end{align*}
implying the first statement. 

From Bayes' rule it follows that $\Pr(x|x) \cdot \Pr (y|y)$ is equal to:
\begin{align}\label{eq_lm_correlation_1}
&\Pr(x|x) \cdot \Pr(y|y) \nonumber \\
&= \frac{\int_{\omega} \Pr(x|\omega)^2 \cdot  p(\omega) d \omega}
	  {\int_{\omega} \Pr(x|\omega) \cdot  p(\omega) d \omega} \cdot 
	   \frac{\int_{\omega} \Pr(y|\omega)^2 \cdot p(\omega) d \omega}
	  {\int_{\omega} \Pr(y|\omega) \cdot p(\omega) d \omega}
\end{align}
Similarly, $\Pr(y|x) \cdot \Pr(x|y)$ is equal to:
\begin{align}\label{eq_lm_correlation_2}
&\Pr(y|x) \cdot \Pr(x|y)  \nonumber \\
&= \frac{\int_{\omega} \Pr(y|\omega) \cdot \Pr(x|\omega) \cdot p(\omega) d \omega}
	  {\int_{\omega} \Pr(x|\omega) \cdot p(\omega) d \omega} \nonumber \\
	  &\cdot 
	   \frac{\int_{\omega} \Pr(x|\omega) \cdot \Pr(y|\omega) \cdot p(\omega) d \omega}
	  {\int_{\omega} \Pr(y|\omega) \cdot p(\omega) d \omega}
\end{align}
Notice that expressions \eqref{eq_lm_correlation_1} and \eqref{eq_lm_correlation_2}
have equal denominators, so we only need to compare nominators.
Let $A$ and $B$ be two random variables such that $A = \Pr(x|\omega)$ 
and $B = \Pr(y|\omega)$. We have:
\begin{align*}
&\sgn \left ( \Pr(x|x) \cdot \Pr(y|y) - \Pr(x|x) \cdot \Pr(y|y) \right ) \\
&= \sgn \left ( \mathbbm E_\omega (A^2) \cdot  \mathbbm E_\omega (B^2) - \mathbbm E_\omega (A \cdot B) ^2 \right ) 
\end{align*}
$\mathbbm E_{\omega}$ is expectation over the distribution $p(\omega)$.
Using the Cauchy-Schwarz inequality 
($\mathbbm E_\omega (A^2) \cdot \mathbbm E_\omega (B^2) \ge |\mathbbm E_\omega (A \cdot B)|^2$,
with strict inequality if $A,B \ne 0$ and $A \ne \lambda \cdot B$),
and the fact that $A$ and $B$ are positive random variables that differ due to stochastic relevance, we obtain
that $\Pr(x|x) \cdot \Pr(y|y) > \Pr(y|x) \cdot \Pr(x|y)$ for $x \ne y$.

The third claim follows analogously. Namely, the nominator of \eqref{eq_lm_correlation_1} is in more general 
form equal to $\mathbbm E_\omega (A_1^2) \cdots  \mathbbm E_\omega (A_{m'}^2) $, while the nominator of  
\eqref{eq_lm_correlation_2} is in more general form equal to $\mathbbm E_\omega (A_1 \cdot A_2) \cdots  \mathbbm E_\omega (A_{m'} \cdot A_1) $,
where $A_i =  \Pr(x_i|\omega)$. By applying the Cauchy-Schwarz inequality, we obtain the claim. 
\end{proof}

\subsection{Proof of Lemma \ref{lm_knowlege_pts}}

\begin{proof}
Consider two arbitrary probability distribution functions $P_1$ and $P_2$, such that: $P_1(x) = P_2(y)$,
$P_1(y) = P_2(x)$, and $P_1(z) = P_2(z)$ for $z \ne x$ and $z \ne y$. PTS with $P_1$ is 
truthful for belief model defined by $Pr_1(z|z) = P_1(z) + (1- P_1(z)) \cdot \beta$, 
$Pr(w|z) =  P_1(w) - P_1(w) \cdot \beta$ , $\forall z, w \in \{0, ..., m-1\}: 
w \ne z$, where $\beta \in (0, 1)$. Similarly, we define posteriors $Pr_2$ based on $P_2$. Notice that
 $Pr_2(w|z) = Pr_1(u(w)|u(z))$,
where $u:\{0, ..., m-1\} \rightarrow \{0, ..., m-1\}$:
\begin{align*}
u(z) = \begin{cases} 
	y &\mbox{ if } z = x\\
	x &\mbox{ if } z = y\\
	z &\mbox{ if } z \ne x \land z \ne y\\
\end{cases}
\end{align*}

Now, suppose that $P_1(x) > P_1(y)$, and set up $\beta$ such that: $Pr_1(x|x) > Pr_1(y|x)$ and
$Pr_1(y|y) < Pr_1(x|y)$. By using the same procedure of proving as in Theorem 1 of \cite{RF:13}, we obtain
that if a mechanism $\tau$ is incentive compatible for both $Pr_1$ and $Pr_2$, 
then:
\begin{align*}
(Pr_1(x|x) - Pr_1(y|x)) \cdot (\Delta(x) - \Delta(y)) &> 0
\\
(Pr_1(y|y) - Pr_1(x|y)) \cdot (\Delta(x) - \Delta(y)) &> 0
\end{align*}  
where $\Delta(x) = \tau(x,x) - \tau(x,y)$ and $\Delta(y) =  \tau(y,x) - \tau(y,y)$. 
The two inequalities, however, contradict by the choice of $Pr_1$. 

Therefore, even though a mechanism might know the size of a region in probability simplex $\mathcal P$
that contains distributions $P$ for which PTS is BNIC, it might not be able to construct a BNIC payment rule. 
\end{proof}

\subsection{Proof of Lemma \ref{lm_PartialTruth_lb}}

\begin{proof}
By Definition \ref{def_lk}, there exists no single mechanism that
incentivizes agents to report honestly if their belief model is arbitrary. Suppose now that 
we have two mechanisms $\tau_1$ and $\tau_2$ that are BNIC under two different 
(complementary) belief models, so that a particular group of agents is
truthful only for one mechanism. We can consider this situation from the perspective 
of a meta-mechanism $\mathcal M$ that has to choose between $\tau_1$ and $\tau_2$.
At a time-step $t$, $\mathcal M$ obtains $k$ reports --- feedback, which, in general, is 
insufficient for determining whether agents lied or not because agents' observations 
are stochastic, while their reports contain noise. Therefore, the problem of choosing 
between $\tau_1$ and $\tau_2$ is an instantiation of a multi-armed bandit problem
(see Section \textit{Mechanism: Adaptive PTS (AdaPTS)}). 
Since, in general, any MAB algorithm pulls suboptimal 
arms $\Omega(\log (N))$ number of times in expectation where $N$ is the total number of pulls (e.g., see \cite{LR:85}), 
we know that meta mechanism $\mathcal M$ will in expectation choose the wrong (untruthful) payments at least $O(\log n)$ times. 
This produces $O(\log n)$ untruthful reports in expectation because non-truthful payments are not truthful for at least one signal value 
and each signal value has strictly positive probability of being endorsed by an agent. 
\end{proof} 

\subsection{Proof of Proposition \ref{prop_truth_exists}}

The proposition is the direct consequence of the two following lemmas. 

\begin{lemma}\label{lm_region_point}
If there exists an informative distribution $P \in \mathcal P$, where $\mathcal P$ is probability 
simplex for $m$-ary signal space, then there exists a region $\mathcal R \subset \mathcal P$ 
such that all $P' \in \mathcal R$ are informative. 
\end{lemma}

\begin{proof}
Since $P$ is informative:
\begin{align*}
\frac{\Pr(x|x)}{P(x)} > \frac{\Pr(y|x)}{P(y)} 
\end{align*}
for all $x \ne y$. The strictness of the inequality implies that
there exists $0 < \delta << \min_z Pr(z)$ such that for all $x \ne y$ we have:
\begin{align*}
\frac{Pr(x|x)}{P(x) + \delta} > \frac{Pr(y|x)}{P(y) - \delta} 
\end{align*}
In other words, we have that any $P' \in \mathcal P$ 
for which $|P'(x) - P(x)| \le \delta$, $x \in \{1, ..., m\}$,
satisfies:
\begin{align*}
\frac{Pr(x|x)}{P'(x)} > \frac{Pr(y|x)}{P'(y)} 
\end{align*}
By putting $\mathcal R = \{ P' | P' \in \mathcal P, \forall x : |P'(x) - P(x)| < \delta \}$,
we obtain the claim. 
\end{proof}

\begin{lemma}\label{lm_point}
There exists an informative distribution $P$. 
\end{lemma}

\begin{proof} 
We only need to show the existence of $P$. Consider a specific signal value $a \in \{0, ..., m-1 \}$
and any other signal value $x \in \{0, ..., m - 1 \}$. Let us define $P$ as:
\begin{align*}
&P(a) = \frac{1}{\alpha} \\
&P(x) = \frac{1}{\alpha} \cdot \min_{y \in \{0, ..., m-1\}} \left [ \frac{\Pr(x|x)}{\Pr(y|x)} \cdot \frac{\Pr(y|y)}{\Pr(a|y)}  \right ] 
\end{align*}
where $\frac{1}{\alpha}$ is a normalization factor so that $\sum_{y \in \{0, ..., m-1\}} P(y) = 1$.
Notice that by Lemma \ref{lm_correlation}:
\begin{align}\label{eq_beta}
 \frac{\Pr(x|x)}{\Pr(y|x)} \cdot \frac{\Pr(y|y)}{\Pr(a|y)} >  \frac{\Pr(x|a)}{\Pr(a|a)}
\end{align}
holds for any $y \in \{0, ..., m-1\}$.  
To prove that $P$ is informative, it is sufficient to show that for any signal values $x \ne z \ne a$ we have:
\begin{align*}
\frac{\Pr(a|a)}{P(a)} > \frac{\Pr(x|a)}{P(x)}\\
\frac{\Pr(x|x)}{P(x)} \ge \frac{\Pr(a|x)}{P(a)}\\
\frac{\Pr(x|x)}{P(x)} > \frac{\Pr(z|x)}{P(z)}
\end{align*} 
Provided that these inequalities hold, the second one can be made strict (while
keeping the other two inequalities strict as well) by reducing all $P(x)$, $x \ne a$, by 
a small enough value, and then re-normalizing $P$. 

The first inequality follows from $\frac{\Pr(a|a)}{P(a)} = \alpha \cdot \Pr(a|a)$
and inequality \eqref{eq_beta}:
\begin{align*}
\frac{\Pr(x|a)}{P(x)} &= \alpha \cdot \frac{\Pr(x|a)}{\min_{y \in \{0, ..., m-1\}} \left [ \frac{\Pr(x|x)}{\Pr(y|x)} \cdot \frac{\Pr(y|y)}{\Pr(a|y)} \right ]} \\
&< \alpha \cdot \frac{\Pr(x|a)}{ \frac{\Pr(x|a)}{\Pr(a|a)}} = \alpha \cdot \Pr(a|a) = \frac{\Pr(a|a)}{P(a)}
\end{align*} 
We obtain the second inequality by putting $y = a$: 
 \begin{align*}
 \frac{\Pr(x|x)}{P(x)} &= \alpha \cdot \frac{\Pr(x|x)}{\min_{y \in \{0, ..., m-1\}} \left [ \frac{\Pr(x|x)}{\Pr(y|x)} \cdot \frac{\Pr(y|y)}{\Pr(a|y)} \right ] } 
 \\&\ge \alpha \cdot \frac{\Pr(x|x)}{\frac{\Pr(x|x)}{\Pr(a|x)} } = \alpha \cdot \Pr(a|x) = \frac{\Pr(a|x)}{P(a)}
\end{align*}
For the third inequality, we use the second inequality
and inequality $\Pr(x|x) \cdot \Pr(z|z) > \Pr(z|x) \cdot \Pr(x|z)$, which follows from Lemma \ref{lm_correlation}. 
We have:
\begin{align*}
 &\frac{\Pr(z|x)}{P(z)} < \frac{\Pr(x|x) \cdot \Pr(z|z)}{\Pr(x|z)} \cdot  \frac{1}{P(z)}\\
 &=\frac{\Pr(x|x) \cdot \Pr(z|z)}{\Pr(x|z)} \cdot \frac{\alpha}{\min_{y \in \{0, ..., m-1\}} \left [ \frac{\Pr(z|z)}{\Pr(y|z)} \cdot \frac{\Pr(y|y)}{\Pr(a|y)} \right ] } 
 \\&\le\alpha \cdot \frac{\Pr(x|x) \cdot \Pr(z|z)}{\Pr(x|z) \cdot \left [ \frac{\Pr(z|z)}{\Pr(x|z)} \cdot \frac{\Pr(x|x)}{\Pr(a|x)} \right ] } = \alpha \cdot Pr(a|x)
 \\&= \frac{\Pr(a|x)}{P(a)} \le \frac{\Pr(x|x)}{P(x)} 
\end{align*}
\end{proof}

\subsection{Proof of Proposition \ref{prop_exist_non_surjective}}

\begin{proof}
One equilibrium non-surjective strategy profile is when agents report value $x$ such that $P(x) = \min_z P(z)$. 
Let us denote by $\tilde Pr$ an agent $i$'s belief regarding the report of her peer agent $j$, i.e., $Y_{j}$, for the considered 
strategy profile. 
Notice that $\tilde Pr(y|x) = \frac{\epsilon}{m} \ge 0$ for $y \ne x$
(due to the reporting noise). However, the strategy profile of reporting $x$ is an equilibrium because 
the expected value of $\hat \epsilon$ is equal to $\epsilon$, so in expectation
$c \cdot \frac{\tilde Pr(y|x)}{P(y)} - d = c \cdot \frac{\epsilon}{m \cdot P(y)} - c \cdot \frac{\epsilon}{m \cdot P(y)}= 0$ for $ y \ne x$, while $c \cdot \frac{\tilde Pr(x|x)}{P(x)} - d > 0$.
That is, an agent's expected payment is strictly maximized when she reports $x$.
\end{proof}

\subsection{Proof of Lemma \ref{lm_disagree}}

\begin{proof}
Since a non-surjective reporting strategy is a non-surjective function of observation $X_i$, we know that $\rho(X_i)$ takes
values in a strict subset of all possible signal values.
The disagreement function is linear in $\mathbbm 1_{Y_i \ne Y_j}$, so it is sufficient to show that $\frac{1}{1-\epsilon} \cdot (\mathbbm 1_{Y_i \ne Y_j} - \frac{\epsilon}{m})$ is in expectation 
greater for truthfulness than for non-surjective reporting strategy profile. In expectation, the expression is equivalent to saying whether two reports of rational agents disagree, 
which is equal to $1 - p_{a}$, where $p_{a}$ is the probability of agreement. The probability of agreement 
in a non-surjective strategy profile is equal to:
\begin{align*}
p_{a,ns} = \sum_{x} \Pr(\{z | \rho(z) = \rho(x)\}, x)  > \sum_x \Pr(x,x) = p_{a,t}
\end{align*}
where the last term is the probability of agreement for truthfulness. Notice that the inequality is strict because in a non-surjective strategy profile there exist
$x$ and $y$ for which $\rho(x) = \rho(y)$, and thus, $\Pr(\{z | \rho(z) = \rho(x) \}, x) \ge \Pr(\{x, y\}, x) =  \Pr(x, x) + \Pr(y, x) > \Pr(x, x)$. 
The last inequality follows from agents' beliefs being fully mixed.   
Since $p_{a,t}$ is strictly smaller than $p_{a,ns}$ for any non-surjective strategy profile, 
we conclude that the disagreement is strictly greater for truthful reporting than for any other non-surjective strategy profile. 
\end{proof}

\subsection{Proof of Theorem \ref{thm_num_pts_dishonest}}

\begin{proof}
Consider a reporting strategy in which agents are honest whenever $P$ is such that truthful reporting is a strict Bayesian-Nash equilibrium of PTS 
(by Proposition \ref{prop_truth_exists}, such $P$ always exists),
and otherwise they use an equilibrium non-surjective strategy profile (which by Proposition \ref{prop_exist_non_surjective} always exists). 
We use the result that the UCB1 algorithm is expected to pull a suboptimal arm $\log(N)$ times, where $N$ is the total number of pulls \cite{A:02}. 
By Lemma \ref{lm_disagree}, the representative $P$ of a truthful reporting region is an optimal arm, while the representative $P$ of 
a non-truthful region is a suboptimal arm. 
Furthermore, the number of pulls in our case corresponds to $n/k$, where $n$ is the total number of agents and $k$ is
the number of agents at time period $t$. 
Since $k << n$ is a fixed parameter, the expected number of lying agents is of the order of $O(\log(n))$.  
\end{proof}

\end{document}